\documentclass[a4paper,11pt]{article}
\usepackage[small]{titlesec}
\usepackage[backend=biber,style=numeric,natbib=false,giveninits=true,doi=true,isbn=false,url=false,date=year]{biblatex}
\DeclareNameAlias{default}{last-first}

\DeclareFieldFormat[article]{title}{#1}
\renewbibmacro{in:}{}
\DeclareFieldFormat[article]{pages}{#1}
\DeclareFieldFormat{journal}{#1}
\renewcommand\bf\bfseries

\renewbibmacro*{volume+number+eid}{%
	\printfield{volume}%
	\setunit*{\addnbspace}
	\printfield{number}%
	\setunit{\addcomma\space}%
	\printfield{eid}}
\DeclareFieldFormat[article]{number}{\mkbibparens{#1}}
\DeclareFieldFormat[article]{volume}{\textbf{#1}}
\DeclareFieldFormat{year}{\mkbibparens{#1}}
\DeclareBibliographyDriver{article}{%
	\printnames{author}:%
	\newunit\newblock
	\printfield{title}%
	\newunit\newblock
	\printfield{journaltitle}
	\newunit
	\iffieldundef{number}{\printfield{volume}}{\printfield{volume}\addspace\printfield{number}}
	\addcomma\addspace\printfield{pages}\addspace
	\printfield{year}
}
\AtEveryBibitem{\clearfield{month}}
\AtEveryBibitem{\clearfield{day}}
\usepackage[english]{babel}
\usepackage[latin1]{inputenc}
\usepackage[T1]{fontenc}
\usepackage{csquotes}
\usepackage{bbm}
\usepackage{amsmath,amsfonts,amsthm,amsbsy,amssymb,dsfont,stmaryrd}
\usepackage[dvipsnames]{xcolor}
\usepackage{braket}

\usepackage{mathtools}

\numberwithin{equation}{section}

\newcommand\myshade{85}
\colorlet{mylinkcolor}{violet}
\colorlet{mycitecolor}{YellowOrange}
\colorlet{myurlcolor}{Aquamarine}

\usepackage[left=2.5cm,right=2.5cm,top=2.5cm,bottom=2.5cm]{geometry}
\usepackage[unicode=true,pdfusetitle,
bookmarks=true,bookmarksnumbered=false,bookmarksopen=false,
breaklinks=false,pdfborder={0 0 1},backref=false,
linkcolor  = mylinkcolor!\myshade!black,
citecolor  = mycitecolor!\myshade!black,
urlcolor   = myurlcolor!\myshade!black,
colorlinks = true,
]
{hyperref}
\usepackage[nameinlink]{cleveref}

\usepackage{tikz}
\usetikzlibrary{arrows}
\usetikzlibrary{intersections}

\usepackage{graphicx}
\usepackage{caption}

\definecolor{ct_black}{HTML}{000000}
\definecolor{ct_orange}{HTML}{ED872D}
\definecolor{ct_purple}{HTML}{7A68A6}
\definecolor{ct_blue}{HTML}{348ABD}
\definecolor{ct_turquoise}{HTML}{188487}
\definecolor{ct_red}{HTML}{E32636}
\definecolor{ct_pink}{HTML}{CF4457}
\definecolor{ct_green}{HTML}{467821}

\definecolor{ct2_green}{HTML}{9FF781}
\definecolor{ct2_green_dark}{HTML}{088A08}

\theoremstyle{plain}
\newtheorem{thm}{\protect\theoremname}[section]
\theoremstyle{plain}
\newtheorem{lem}[thm]{\protect\lemmaname}
\theoremstyle{plain}

\theoremstyle{plain}
\newtheorem{prop}[thm]{\protect\propositionname}
\theoremstyle{remark}

\theoremstyle{remark}
\newtheorem{rem}[thm]{\protect\remarkname}
\theoremstyle{definition}
\newtheorem{defn}[thm]{\protect\definitionname}
\theoremstyle{plain}

  \providecommand{\assumptionname}{Assumption}
\providecommand{\claimname}{Claim}
\providecommand{\corollaryname}{Corollary}
\providecommand{\definitionname}{Definition}
\providecommand{\lemmaname}{Lemma}
\providecommand{\propositionname}{Proposition}
\providecommand{\remarkname}{Remark}
\providecommand{\theoremname}{Theorem}
\providecommand{\examplename}{Example}

\crefname{section}{Section}{Sections}
\crefname{appendix}{Appendix}{Appendices}
\crefname{figure}{Figure}{Figures}
\crefname{assumption}{Assumption}{Assumptions}
\crefname{thm}{Theorem}{Theorems}
\crefname{lem}{Lemma}{Lemmas}
\crefformat{equation}{(#2#1#3)}

\crefrangelabelformat{equation}{(#3#1#4--#5#2#6)}

\crefmultiformat{equation}{(#2#1#3}{, #2#1#3)}{#2#1#3}{#2#1#3}
\Crefmultiformat{equation}{(#2#1#3}{, #2#1#3)}{#2#1#3}{#2#1#3}

\newtheorem*{lem*}{\protect\lemmaname}

\newcommand{\ee}{\operatorname{e}}
\newcommand{\ii}{\operatorname{i}}
\newcommand{\Mat}{\operatorname{Mat}}
\newcommand{\ZZ}{\mathbb{Z}}
\newcommand{\TT}{\mathbb{T}}

\newcommand{\NN}{\mathbb{N}}
\newcommand{\RR}{\mathbb{R}}
\newcommand{\CC}{\mathbb{C}}

\newcommand{\EE}{\mathbb{E}}

\newcommand{\Gras}[2]{\operatorname{Gr}_{#1}(#2)}

\newcommand{\norm}[1]{\|#1\|}

\newcommand{\dif}{\operatorname{d}}
\newcommand{\tr}{\operatorname{tr}}
\renewcommand{\Im}{\operatorname{\mathbb{I}\mathbbm{m}}}
\renewcommand{\Re}{\operatorname{\mathbb{R}\mathbbm{e}}}
\newcommand{\ve}{\varepsilon}
\newcommand{\Id}{\mathds{1}}
\newcommand{\HH}{\mathcal{H}}

\newcommand{\BLO}[1]{\mathcal{B}(#1)}
\newcommand{\Open}[1]{\mathrm{Open}(#1)}
\newcommand{\dist}[1]{\mathrm{dist}(#1)}

\newcommand{\Chern}{\operatorname{Ch}}

\newcommand{\supp}{\operatorname{supp}}

\newcommand{\im}{\operatorname{im}}

\newcommand{\vf}{\varphi}

\title{The Topology of Mobility-Gapped Insulators}
\author{Jacob Shapiro\\
\footnotesize{Mathematics Department, Columbia University, New York, NY 10027,
	USA}}

\begin{document}
	
\maketitle

\begin{abstract}
Studying deterministic operators, we define an appropriate topology on the space of mobility-gapped insulators such that topological invariants are continuous maps into discrete spaces, we prove that this is indeed the case for the integer quantum Hall effect, and lastly we show why our "insulator" condition makes sense from the point of view of the localization theory using the fractional moments method.
\end{abstract}

\section{Introduction}
Topological insulators \cite{Hasan_Kane_2010} are usually studied in physics by assuming translation invariance, which allows for a topological description of Hamiltonians in terms of continuous maps from the Brillouin torus $\TT^d\to X$ where $X$ is some smooth manifold which depends on the symmetry class under consideration (for example, for a system in class A (no symmetry) gapped after $n$ levels, $X=\Gras{n}{\CC^\infty}$, the Grassmannian manifold). Such a description is extremely convenient because one may immediately apply classical results from algebraic topology, for example, that the set of homotopy classes $[\TT^2\to\Gras{n}{\CC^\infty}]\cong\ZZ$. This approach led to many classification results \cite{Schnyder_Ryu_Furusaki_Ludwig_1367-2630-12-6-065010,Schnyder_Ryu_Furusaki_Ludwig_PhysRevB.78.195125} which culminated in the Kitaev periodic table of topological insulators \cite{Kitaev2009}, all the while ignoring the fact that the systems to be analyzed are actually \emph{not} translation invariant, and in fact certain physical features of the phenomena demand strong-disorder. Indeed, the plateaus of the integer quantum Hall effect (IQHE henceforth) are explained only when assuming the Fermi energy lies in a region of localized states (\emph{the mobility gap regime}) which cannot appear in a translation invariant system.

Hence a physically more realistic description calls for understanding \emph{disordered} systems in which Bloch decomposition cannot be used. This has been done for the IQHE in \cite{Bellissard_1994} by applying ideas from non-commutative geometry, and later generalized to the Kitaev table in \cite{PSB_2016,Bourne2017} (and references therein). One problem with the application of non-commutative geometry is that it still required the Fermi energy to be placed in a spectral gap, which is why \cite{Bellissard_1994} goes beyond the C-star algebra generated by continuous functions of the Hamiltonian by defining a so-called "non-commutative Sobolev spaces". Such an approach still uses crucially the translation invariance of the system; in contrast to the studies in physics, however, translation is used in terms of the probability distributions defining a random model. That is, a whole statistical ensemble of models is considered simultaneously, proofs use the covariance property, and statements are made either almost-surely or about disorder averages. Furthermore sometimes certain statements could also not be extended to these bigger Sobolev spaces, including the bulk-edge correspondence or the definition of edge invariants (which requires a certain regularization, see \cite[Eq-n (1.2)]{EGS_2005}).

Further explorations of the mobility gap regime using only one particular deterministic realization (i.e., \emph{without} referring to a statistical ensemble) and without using covariance were pioneered in \cite{EGS_2005} for the IQHE and extended in \cite{Graf_Shapiro_2018_1D_Chiral_BEC,Shapiro_Tauber_BEC} for chiral and Floquet topological systems respectively. These studies demonstrate that topological properties do not need translation invariance nor statistical averaging, and should be associated to one particular mobility gapped Hamiltonian rather than an ensemble. Since \cite{EGS_2005,Graf_Shapiro_2018_1D_Chiral_BEC,Shapiro_Tauber_BEC} do not use the algebraic framework of non-commutative geometry, topological properties such as defining the ambient space of topological insulators or establishing local constancy (w.r.t. deformations of Hamiltonians) of the invariants have been hard to establish (in contrast to the very natural appearance of such properties in the framework of \cite{Bellissard_1994}), despite the fact they are important, being one of the defining physical features of topological insulators.

In this note, we have the modest goal of continuing the deterministic line of research of \cite{EGS_2005,Graf_Shapiro_2018_1D_Chiral_BEC,Shapiro_Tauber_BEC} by defining a space of insulators and proving that a topological invariant is locally constant with respect to it. We use the IQHE as a case study since we understand it best, though future studies for other cases of topological insulators are certainly interesting. For example, chiral one-dimensional systems may relate to an extension of Fredholm theory to operators without closed range, but which obey a localization estimate instead, which allows to salvage the local constancy of an integer-valued index. 

This note is organized as follows. We begin by giving precise definition of what we mean by a "topological insulator" in our deterministic setting, define a topology on this space, and discuss its various properties. In the next section we take up the IQHE as a case study and establish the deterministic local constancy. In the last section we discuss why our deterministic definition of "insulator" makes sense for probabilistic models which exhibit localization about the Fermi energy.

In regards to existing literature, the question of the appropriate topology for topological insulators has already been raised in previous papers, see the discussion in \cite[Appendix D]{Freed2013} and \cite[Introduction]{Kellendonk2017}. The mobility gap continuity has also been dealt with before, in the context of probabilistic covariant models, see \cite[Proposition 5.2]{PRODAN20161150} and references therein.

\section{Deterministic topological insulators}
Let $d\in\NN$, the space dimension, $N\in\NN$, the internal number of degrees of freedom, be given and fixed. We define our Hilbert space as $\HH:=\ell^2(\ZZ^d)\otimes\CC^N$, and for an operator $A\in\BLO{\HH}$, $A_{xy}\equiv\langle\delta_x, A\delta_y\rangle$ is an $N\times N$ matrix with $(\delta_x)_{x\in\ZZ^d}$ being the position basis of $\ell^2(\ZZ^d)$. $\norm{\cdot}$ is either a matrix norm on $\Mat_N(\CC)$ or the $1$-norm on $\RR^d$.

We next define a metric on $\BLO{\HH}$. For brevity, let $X:=(0,\infty)$.

\begin{defn}For any $A,B\in\BLO{\HH}$, define the \emph{local distance} between them as \begin{align}
	d_\ell(A,B) := \inf(\Set{t\in X | \exists C,\mu\in X : t = \max(\Set{C,\mu^{-1}})\land\norm{(A-B)_{xy}} \leq C \ee^{-\mu\norm{x-y}}\forall x,y\in\ZZ^d })
	\label{eq:local metric}\,.
\end{align}
This distance measures not only how close-by position-basis matrix elements are, but also their rate of off-diagonal exponential decay. The metric is mainly used as follows:
\begin{lem}
	If $d_\ell(A,B)\in X$ then $\norm{(A-B)_{xy}} \leq d_\ell(A,B) \ee^{-\norm{x-y}/d_\ell(A,B)}\forall x,y\in\ZZ^d$.
	\label{lem:meaning of local norm for approximating exp decay}
\end{lem}
\begin{proof}
	By the approximation property of the infimum, we have $\forall\ve>0$ some $C_\ve,\mu_\ve \in X$ such that $d_\ell(A,B)\leq \max(C_\ve,\mu_\ve^{-1}) < d_\ell(A,B)+\ve$ and $\norm{(A-B)_{xy}} \leq C_\ve \ee^{-\mu_\ve\norm{x-y}}\forall x,y\in\ZZ^d$. This implies that $C_\ve \leq d_\ell(A,B)+\ve$ and $\mu_\ve \geq (d_\ell(A,B)+\ve)^{-1}$. Hence \begin{align}
		\norm{(A-B)_{xy}} \leq (d_\ell(A,B)+\ve) \ee^{-\norm{x-y}/(d_\ell(A,B)+\ve)}
	\end{align}
	for all $\ve>0$ which implies the result.
\end{proof}
\end{defn}
It is of course comforting to know that
\begin{lem}
	$d_\ell$ is a metric on $\BLO{\HH}$. 
	\begin{proof}
		By definition it is obvious that $d_\ell$ is non-negative and symmetric.
		
		Next, we want that $d_\ell(A,A)=0$. The second requirement in the set under the infimum in \cref{eq:local metric} becomes trivial for $d_\ell(A,A)=0$, so that taking $C>0$ arbitrarily small and $\mu$ fixed for example we reach a zero infimum.
		
		If we assume that $d_\ell(A,B)=0$, we want that $A=B$. This assumption implies that $\forall\ve>0\exists C_\ve,\mu_\ve\in X$ with $\max(\Set{C_\ve,\mu_\ve^{-1}})<\ve$ and $\norm{(A-B)_{xy}} \leq C_\ve \ee^{-\mu_\ve\norm{x-y}}$, which implies $\norm{(A-B)_{xy}} \leq C \leq \ve$. This means $A=B$.
		
		Finally we get to the triangle inequality. Let $A,B,C\in\BLO{\HH}$ be given. We have by the usual triangle inequality for the matrix norm that $\norm{(A-B)_{xy}}\leq\norm{(A-C)_{xy}} +\norm{(C-B)_{xy}} $. By \cref{lem:meaning of local norm for approximating exp decay}, we then get \begin{align*}\norm{(A-B)_{xy}} &\leq d_\ell(A,C) \ee^{-\norm{x-y}/d_\ell(A,C)} + d_\ell(C,B) \ee^{-\norm{x-y}/d_\ell(C,B)} \\ &\leq (d_\ell(A,C)  + d_\ell(C,B) )\ee^{-\norm{x-y}/(\min(d_\ell(A,C),d_\ell(C,B )))}\end{align*} which means that $d_\ell(A,B)\leq\max(\Set{d_\ell(A,C)  + d_\ell(C,B),\min(d_\ell(A,C),d_\ell(C,B ))}) = d_\ell(A,C)  + d_\ell(C,B)$, so we are finished.
	\end{proof}
\end{lem}

\begin{rem}
	$d_\ell$ is unfortunately not homogeneous, so it cannot induce a norm (it is translation invariant though). This is because it measures also the rate of exponential decay. Compare this with the local norm of \cite{GrafTauber18} which has a fixed rate of decay.
\end{rem}

\begin{rem}\label{rem:the local metric topology is too fine}
	To require off-diagonal \emph{exponential} decay is probably stronger than what is necessary for the topological invariants to be well-defined and continuous (as we see below). However, in the interest of keeping the calculations somewhat simpler, we prefer to stipulate one concrete form of off-diagonal decay. This means that the topology induced by $d_\ell$ on the space insulating Hamiltonians (to be defined below) is finer than the \emph{initial topology} corresponding to the topological invariant considered as a map from insulating Hamiltonians into $\ZZ$. Since a continuous map may stop being so if the topology of its domain is made coarser, that means that using the topology induced by $d_\ell$, a-priori, we might not detect all path-connected components on the space of topological insulators. Since in this paper we anyway don't concern ourselves with calculating the space of path-connected components of topological insulators, we ignore this problem.
\end{rem}

\begin{lem}
	If $A_n\to A$ in the topology induced by $d_\ell$ then $A_n\to A$ in the norm operator topology.
\end{lem}
\begin{proof}
	Recall Holmgren's bound \begin{align*}
		\norm{A} \leq \max_{x\leftrightarrow y}\sup_{y\in\ZZ^d}\sum_{x\in\ZZ^d}\norm{A_{xy}}
	\end{align*} 
	which together with \cref{lem:meaning of local norm for approximating exp decay} implies that \begin{align}\norm{A_n-A}\leq d_l(A_n,A) (\coth(\frac{1}{2d_l(A_n,A)}))^d\,.\label{eq:estimate of operator norm by local metric}\end{align}
	
	We conclude by noting that $\coth(t)\to1$ as $t\to\infty$.
\end{proof}

This shows also that $A_n\to A$ in $d_\ell$ does \emph{not} imply the same in trace norm.

\begin{defn}
	A Hamiltonian $H$ is a self-adjoint operator in $\BLO{\HH}$ such that $d_\ell(0,H)<\infty$.
\end{defn}

Without loss of generality we assume onwards that the Fermi energy is placed at zero energy (if this is not the case replace the Hamiltonian by a shifted one).

\begin{defn}\label{def:insulator}
	An insulator is a Hamiltonian $H$ such that there is some open interval $\RR\supseteq\Delta\ni 0$ such that: \begin{itemize}
			\item All eigenvalues of $H$ within $\Delta$ are \emph{uniformly} of finite degeneracy: $$\sup_{\lambda\in\Delta}\|\chi_{\{\lambda\}}(H)\|_1 < \infty\,.$$
			\item With $B_1(\Delta)$ the set of all Borel bounded functions $f:\RR\to\CC$ which obey $|f(\lambda)|\leq 1\forall\lambda\in\RR$ and which are constant below and above $\Delta$ (possibly with different constants) (see \cite{EGS_2005}) there are $C<\infty,\mu>0,a\in\ell^p(\ZZ^d)$ such that  \begin{align}
				\sup_{f\in B_1(\Delta)} \norm{f(H)_{xy}} \leq C |a(x)|^{-1} \ee^{-\mu\norm{x-y}}\qquad(x,y\in\ZZ^d)\,.
				\label{eq:dynamical constraint for an insulator}
			\end{align}
		\end{itemize}
The space of all insulators where the objects involved $\Delta,C,a,\mu$ in the estimates \cref{eq:dynamical constraint for an insulator} are uniformly bounded by some fixed given worst objects $\Delta_0,C_0,a_0,\mu_0$ is denoted by $\mathcal{I}\equiv\mathcal{I}(\Delta_0,C_0,a_0,\mu_0)$, i.e., we have $|\Delta_0|>0$, $C_0<\infty$, $a_0\in\ell^1(\ZZ^d)$, $\mu_0>0$ and and for any insulator $H\in\mathcal{I}(\Delta_0,C_0,a_0,\mu_0)$ we have: $|\Delta|\geq|\Delta_0|$, $C\leq C_0$, $\|a\|_1\leq \|a_0\|_1$ and $\mu\geq\mu_0$. $\mathcal{I}$ is a subset of $\BLO{\HH}$. We give it the subspace topology induced by the metric topology from $d_\ell$ on the space of all Hamiltonians.
\end{defn}

This definition is not new. It encompasses essentially the same constraints as in \cite[eq-ns (1.1)-(1.3)]{EGS_2005}, has been used already in \cite{Graf_Shapiro_2018_1D_Chiral_BEC,Shapiro_Tauber_BEC}, and represents almost-sure consequences of a probabilistic model that exhibits strong dynamical localization (in the sense of \cite{Aizenman_Graf_1998}) about zero energy. See further remarks in \cref{sec:signatures of localization}.

We note that it is probably false that $\mathcal{I}$ is an open subset with respect to the topology induced by $d_\ell$. Indeed, in \cite{gordon1994} it is observed that Anderson localization of "generic" models (sufficiently random) breaks down by a rank-one perturbation with arbitrarily small norm. See also \cite{delrio1994}.

One could hope to get rid of the uniform objects $\Delta_0,C_0,a_0,\mu_0$ on which $\mathcal{I}$ depends. This might be possible, but probably requires further specification of the details of the models considered and how their randomness arises (i.e. the probability distributions). In order to avoid this specificity we use these rather unsatisfying uniform bounds. They allow us to conclude that if we have a sequence $(H_n)_n\subset\mathcal{I}$ such that $d_\ell(H_n,H)\to0$ for some fixed $H\in\mathcal{I}$, the corresponding objects in the localization estimates $\Delta_n,C_n,a_n,\mu_n$ cannot explode. In analogy to the spectral gap regime, this is tantamount to assuming not only that there is a gap, but the whole collection of operators we consider has a uniform gap, which is unnecessary for spectrally gapped systems as the resolvent set is open in $\CC$. Stated differently, the set of invertible operators is open in $\BLO{\HH}$. In turn, this uniform restriction means we cannot detect the path-connected components of topological phases of insulators.

Later on, we prove that 
\begin{lem}\label{lem:Energy avg of Greens function decays uniformly in height from real axis} If $H$ is an insulator as above, then for any compact sub-interval $\Delta'\subsetneq\Delta$, there is some $s\in(0,1),b\in\ell^1(\ZZ^d)$ such that for all $\alpha\in\NN$ we have some $D_{\alpha}<\infty$ with \begin{align}
	\sup_{\eta\neq0}\int_{\Delta'}\norm{G(x,y;\cdot+\ii\eta)}^s \leq D_{\alpha} |b(x)|^{-1}(1+\|x-y\|)^{-\alpha}\qquad(x,y\in\ZZ^d)
	\label{eq:additional constraint for topological insulators}
	\end{align} where $G(x,y;z)\equiv R(z)_{xy} \equiv (H-z\Id)^{-1}_{xy}$ is the Greens' function associated to $H$.
\end{lem}

While this is anyway an almost-sure consequence of the probabilistic fractional moment condition (see \cref{prop:Energy avg of Greens function decays uniformly in height from real axis}), one can prove this directly from \cref{def:insulator}. We will crucially use \cref{lem:Energy avg of Greens function decays uniformly in height from real axis} to measure the proximity of spectral projections of two insulators.

\subsection{The integer quantum Hall effect}

We take the integer quantum Hall effect as a convenient case study for topological insulators, since much is known about it already; we rely mainly on \cite{EGS_2005}. For the IQHE, one takes $d=2$, and the topological invariant, physically the transversal (Hall) conductivity, is given by the Kubo formula as the Chern number \begin{align}
	\mathcal{I}\ni H \mapsto \Chern(H) := 2 \pi \ii \tr \ve_{\alpha\beta} P P_{,\alpha} P_{,\beta}\in\ZZ\,.
	\label{eq:the Chern number}
\end{align}
Here $\ve_{\alpha\beta}$ is the anti-symmetric tensor (with Einstein summation), $A_{,\alpha} := -\ii[\Lambda_\alpha,A]$ is the non-commutative derivative of an operator $A$ with $\Lambda_\alpha := \Lambda(X_\alpha)$, $X_\alpha$ the position operator in direction $\alpha$ on $\ell^2(\ZZ^2)$, and $\Lambda:\ZZ\to\RR$ is any \emph{switch function} (in the sense of \cite{Elbau_Graf_2002}, that is, any measurable interpolation from zero on negative values to one on positive values with bounded variation)--the choice of $\Lambda$ does not influence the value of $\Chern$, but is fixed once and for all. Finally, $P:=\chi_{(-\infty,0)}(H)$ is the Fermi projection associated with $H$. \cite{EGS_2005} contains a proof of the fact that \cref{eq:the Chern number} is a well defined map.

Our main result is

\begin{thm}
	The map $\Chern:\mathcal{I}\to\ZZ$ is continuous.
	\label{thm:Chern number is continuous}
\end{thm}
Consequently, since $\Chern$ is $\ZZ$-valued, it is locally constant. This finally gives a concrete criterion to be able to tell when two Hamiltonians will have the same Chern number, without having to actually calculate it. As noted in \cref{rem:the local metric topology is too fine}, it is not the weakest possible criterion.

Since the Chern number of a Hamiltonian is defined through its associated Fermi projection, one would naively hope to bound $d_\ell(P,P')$ by $d_\ell(H,H')$. Not only does this turn out not to work, but the very definition of an insulator shows that we can't even hope to have $\norm{P_{xy}}$ decaying in $\norm{x-y}$ uniformly. It is even false if we relax the condition to merely asking that $\norm{(P-P')_{xy}}$ is small, and has some off-diagonal decay and diagonal blow-up. Indeed, the problem is that $\chi_{(-\infty,0)}$ is not a continuous function, and we are considering operators that precisely have spectrum near zero, so even considering just the diagonal element $\norm{(P-P')_{xx}}$ for fixed $x$, an arbitrarily small change from $H$ to $H'$ could make an eigenvalue jump over zero so that $\norm{(P-P')_{xx}}$ is one.

The way out is to mimic the probabilistic approach (which cures things by looking at averages, which has the effect of smoothening discontinuities), with a trick of averaging over the Fermi energy within the gap. Indeed, the point is that even though $\chi_{(-\infty,0)}$ is not a continuous function, $\int \chi_{(-\infty,\lambda)}\dif{\lambda}$ \emph{is}. Our main effort below is to make this intuitive argument rigorous.

The Fermi-energy averaging is permitted using the following key result:
\begin{prop}(\cite[Proposition 2]{EGS_2005}) Let $H\in\mathcal{I}$. Then according to \cref{def:insulator}, there is some $\Delta\in\Open{\RR}$ such that $0\in\Delta$ and for which the mobility gap estimates are fulfilled. Then the following map is constant: $$\Delta\ni E_F\mapsto\Chern(H-E_F\Id)\in\ZZ\,.$$
	\label{prop:EGS_Chern_const_in_mob_gap}
\end{prop}
We can then formulate precisely in which sense is $P-P'$ small given that $H-H'$ is small:
\begin{prop}
	Let $H,H'\in\mathcal{I}$ with $d_\ell(H,H')<\infty$. Let $\Delta'$ be a compact interval contained in the localization estimate interval (\cref{def:insulator}) of both $H$ and $H'$. Define $$P_{\lambda} := \chi_{(-\infty,0)}(H-\lambda\Id)=\chi_{(-\infty,\lambda)}(H)\,.$$ Then we have some $s\in(0,1)$, $C<\infty$, $a\in\ell^1(\ZZ^d)$ (dependent on $H,H'$) such that for all $\alpha\in\NN$, \begin{align}\int_{\Delta'}\norm{(P_{\lambda}-P'_{\lambda})_{xy}}\dif{\lambda}\leq C d_\ell(H,H')^{s}|a(x)|^{-1}(1+\norm{x-y})^{-\alpha}\qquad(x,y\in\ZZ^d)\,.\end{align}
	\label{prop:continuity of Fermi projections on average}
\end{prop}
\begin{proof}
	We replace the disorder averaging of \cite[Proposition 5.2]{PRODAN20161150} with Fermi-energy averaging. We start from the formula \cite{Aizenman_Graf_1998} $$ P_\lambda = \frac{\ii}{2\pi}\int_{\Gamma(\lambda)} R(z)\dif{z} $$ where $\Gamma(\lambda)$ is a rectangular curve in $\CC$ going counter-clockwise passing the points $\lambda+\ii,\lambda-\ii,-\norm{H}-1-\ii,-\norm{H}-1+\ii$. We divide the curve into two parts: $\Gamma_1(\lambda)$ which is the two horizontal segments and the left vertical segment, and $\Gamma_2(\lambda)$, the right vertical segment. On $\Gamma_1(\lambda)$, $z$ is always a minimum distance of 1 from $\sigma(H)$ so that one may use the Combes-Thomas estimate. On $\Gamma_2(\lambda)$ we must use localization, since we (possibly) cross the spectrum as we pass the real axis.
	
	We thus find: \begin{align*}
		\norm{(P_{\lambda}-P'_{\lambda})_{xy}} &= \norm{(\frac{\ii}{2\pi}\int_{\Gamma(\lambda)} R(z)\dif{z}-\frac{\ii}{2\pi}\int_{\Gamma(\lambda)} R'(z)\dif{z})_{xy}} \\
		&\leq \norm{\frac{\ii}{2\pi}\int_{\Gamma_1(\lambda)} (G(x,y;z)- G'(x,y;z))\dif{z}}+\norm{\frac{\ii}{2\pi}\int_{\Gamma_2(\lambda)} (G(x,y;z)- G'(x,y;z))\dif{z}}
	\end{align*}
	Then with the resolvent identity \begin{align*}
		\norm{\frac{\ii}{2\pi}\int_{\Gamma_1(\lambda)} (G(x,y;z)- G'(x,y;z))\dif{z}} &\leq \frac{1}{2\pi}\int_{\Gamma_1(\lambda)} \norm{G(x,y;z)- G'(x,y;z)}|\dif{z}| \\
		&\leq \frac{1}{2\pi}\int_{\Gamma_1(\lambda)} \sum_{x',x''}\norm{G(x,x';z)}\norm{(H'-H)_{x',x''}}\norm{G'(x'',y;z)}|\dif{z}|
	\end{align*}
	Using \cref{lem:meaning of local norm for approximating exp decay} and the Combes-Thomas estimate (for some universal $\mu>0$) $$\norm{G(x,y;z)}\leq \frac{2}{\dist{z,\sigma(H)}}\ee^{-\mu \dist{z,\sigma(H)} \norm{x-y} }$$ we now estimate this (recalling that for $z\in\Gamma_1(\lambda)$, $\dist{z,\sigma(H)}\geq 1$) \begin{align*}
		\dots & \leq  \frac{2}{\pi}\int_{\Gamma_1(\lambda)} \sum_{x',x''}\ee^{-\mu\norm{x-x'}}d_\ell(H,H')\ee^{-\norm{x'-x''}/d_\ell(H,H')}\ee^{-\mu\norm{x''-y}}|\dif{z}| \\
		& \leq \frac{2}{\pi} |\Gamma_1(\lambda)| (\coth(\mu/2))^{2d} d_\ell(H,H')  \ee^{-\frac{1}{2}\min(\mu,d_\ell(H,H')^{-1})\norm{x-y}}
 	\end{align*}
 	We see that $\Gamma_1(\lambda)$ apparently doesn't require the averaging over energy.
 	
 	On the other hand for $\Gamma_2(\lambda)$, we do use the localization estimate, which needs the Fermi energy averaging; this is where \cref{lem:Energy avg of Greens function decays uniformly in height from real axis} comes in. Let $s\in(0,1)$. Then using the basic estimate $\norm{G(x,y;\lambda+\ii\eta)}\leq |\eta|^{-1}$ we find
 	\begin{align*}
 		& \int_{\Delta'}\norm{\frac{\ii}{2\pi}\int_{\Gamma_2(\lambda)} (G(x,y;z)- G'(x,y;z))\dif{z}}\dif{\lambda}=\\ &= \int_{\Delta'}\norm{\frac{-1}{2\pi}\int_{-1}^{1} (G(x,y;\lambda+\ii\eta)- G'(x,y;\lambda+\ii\eta))\dif{\eta}}\dif{\lambda} \\
 		&\leq \frac{1}{2\pi}\int_{-1}^{1}\int_{\Delta'} \norm{G(x,y;\lambda+\ii\eta)- G'(x,y;\lambda+\ii\eta)}\dif{\lambda}\dif{\eta} \\ 
 		&\leq \frac{1}{2\pi}\int_{-1}^{1}\int_{\Delta'} |\frac{2}{\eta}|^{1-s/2}\norm{G(x,y;\lambda+\ii\eta)- G'(x,y;\lambda+\ii\eta)}^{s/2}\dif{\lambda}\dif{\eta}
 	\end{align*}
 	Only now, after pulling a fractional power of the imaginary energy (unlike in \cite{PRODAN20161150}), do we use the resolvent identity: 
 	\begin{align*}
 		\dots &\leq \frac{1}{2\pi}\int_{-1}^{1}\int_{\Delta'} |\frac{2}{\eta}|^{1-s/2}\sum_{x',x''}\norm{G(x,x';\lambda+\ii\eta)}^{s/2}\norm{(H'-H)_{x',x''}}^{s/2}\norm{ G'(x'',y;\lambda+\ii\eta)}^{s/2}\dif{\lambda}\dif{\eta}
 	\end{align*}
 	We may pull out the $\sum_{x',x''}$ sum out of the integrals using Fatou. Use the Cauchy-Schwarz inequality on the $\int_{\Delta'}\cdot\dif{\lambda}$ integral to get  	\begin{align*}
 	\dots &\leq \frac{1}{2\pi}\sum_{x',x''}\int_{-1}^{1} |\frac{2}{\eta}|^{1-s/2}\dif{\eta}\norm{(H'-H)_{x',x''}}^{s/2}(\sup_{\eta\neq 0}\int_{\Delta'}\norm{G(x,x';\cdot+\ii\eta)}^{s})^{1/2}\times\\&\quad\times(\sup_{\eta\neq0}\int_{\Delta'}\norm{ G'(x'',y;\cdot+\ii\eta)}^{s})^{1/2}
 	\end{align*}
 	
 	At this point we employ the part of the assumption on $H,H'$ concerning \cref{eq:additional constraint for topological insulators}, so using $s>0$, we have for any $\alpha\in\NN$:
 	\begin{align*}
 	\dots &\leq \frac{2^{2-s/2}D_{0,\alpha}}{\pi s}d_\ell(H,H')^{s/2} \sum_{x',x''}  \ee^{-\frac{s}{2d_\ell(H,H')}\norm{x'-x''}} |b(x)|^{-1/2}(1+\norm{x-x'})^{-\alpha/2}|b(y)|^{-1/2}(1+\norm{x''-y})^{-\alpha/2}\,.
 	\end{align*}
 	
 	Using the triangle inequality $(1+\norm{x-y})^{-\alpha}(1+\norm{y-z})^{-\alpha} \leq (1+\norm{x-z})^{-\alpha}$ we can pull out a polynomially-decaying factor in $\norm{x-y}$ to get: 
 	
 	\begin{align*}
 	\dots &\leq \frac{2^{2-s/2}D_{0,\alpha} Q}{\pi s}d_\ell(H,H')^{s/2} |b(y)|^{-1/2}(1+\norm{x-y})^{-\alpha/2} 
 	\end{align*}
 	with $$ Q := \sum_{x',x''}  \ee^{-\frac{s}{2d_\ell(H,H')}\norm{x'-x''}} |b(x)|^{-1/2}(1+\norm{x-x'})^{-\alpha/2}|b(y)|^{-1/2}(1+\norm{x''-y})^{-\alpha/2}\,.$$
 	
 	But $Q<\infty$ manifestly, so we get our result.
\end{proof}

\begin{proof}[Proof of \cref{thm:Chern number is continuous}]
	Let $H\in\mathcal{I}$ be given. We seek some $\ve>0$ (dependent on $H$) such that if $H'\in \mathcal{I}$ with $d_\ell(H,H')<\ve$, then $\Chern(H)=\Chern(H')$.
	
	By \cref{prop:EGS_Chern_const_in_mob_gap} we may replace the Chern number with its average within the mobility gap to get
	\begin{align*}
	&|\Chern(H)-\Chern(H')|\leq\\ 
	& \leq\frac{1}{|\Delta'|}|\int_{\Delta'}\Chern(H-\lambda\Id)\dif{\lambda}-\int_{\Delta'}\Chern(H'-\lambda\Id)\dif{\lambda}|\\
	& \leq \frac{2\pi}{|\Delta'|}\int_{\Delta'}|\tr(\ve_{\alpha\beta}(P_\lambda P_{\lambda,\alpha} P_{\lambda,\beta}-P'_\lambda P'_{\lambda,\alpha} P'_{\lambda,\beta}))|\dif{\lambda} \\
	& \leq \frac{2\pi}{|\Delta'|} \sum_{(\alpha,\beta) = (1,2),(2,1)} \int_{\Delta'}\norm{(P_\lambda-P'_\lambda)P_{\lambda,\alpha} P_{\lambda,\beta})}_1 + \norm{P'_\lambda(P_\lambda-P'_\lambda)_{,\alpha} P_{\lambda,\beta})}_1 + \norm{P'_\lambda P'_{\lambda,\alpha} (P_\lambda-P'_\lambda)_{,\beta}}_1\dif{\lambda}\,.
	\end{align*}
	
	We will use the estimate \begin{align*}
	\norm{AB}_1 \leq \sum_{xyz} \norm{A_{xy}}\norm{B_{yz}}\,.
	\end{align*}
	
	Consider the first term with $\alpha\neq\beta$: \begin{align*}
	\int_{\Delta'}\norm{(P_\lambda-P'_\lambda)P_{\lambda,\alpha} P_{\lambda,\beta})}_1 \dif{\lambda} & \leq \int_{\Delta'} \sum_{xyz}\norm{(P_\lambda-P'_\lambda)_{xy}}\norm{(P_{\lambda,\alpha} P_{\lambda,\beta})_{yz}}\dif{\lambda}\,.
	\end{align*}
	
	Using \cref{eq:dynamical constraint for an insulator} we know that $P_\lambda$ have off-diagonal decay with diagonal explosion, which was called "weakly-local" in \cite{Shapiro_Tauber_BEC}. Since $\chi_{(-\infty,\lambda)}\in B_1(\Delta)$, the "weakly-local" estimate we get does not depend on $\lambda$. Using \cite[Remark 3.4]{Shapiro_Tauber_BEC} we estimate $\norm{(P_{\lambda,\alpha} P_{\lambda,\beta})_{yz}}\leq C (1+\norm{y-z})^{-\alpha}(1+\norm{z})^{-\alpha}$ for some $\alpha\in\NN$ as large as we want, and the constant $C$ does not depend on $\lambda$. We conclude now using \cref{prop:continuity of Fermi projections on average} that 
	
	\begin{align*}
	\int_{\Delta'}\norm{(P_\lambda-P'_\lambda)P_{\lambda,\alpha} P_{\lambda,\beta})}_1 \dif{\lambda} & \leq \sum_{xyz}C (1+\norm{y-z})^{-\alpha}(1+\norm{z})^{-\alpha}\int_{\Delta'} \norm{(P_\lambda-P'_\lambda)_{xy}}\dif{\lambda} \\ 
	&\leq \sum_{xyz}C (1+\norm{y-z})^{-\alpha}(1+\norm{z})^{-\alpha} \times \\
	&\qquad\times \frac{2^{2-s/2}D_{0,\alpha} Q}{\pi s}d_\ell(H,H')^{s/2} |b(y)|^{-1/2}(1+\norm{x-y})^{-\alpha/2} \,.
	\end{align*}
	
	This last expression after the triple sum is summable so that we get some constant times $d_\ell(H,H')^{s/2}$, which means we can make this term as small as we like by appropriate choice of $H'$.

	Consider now one of the derivative terms with $\alpha\neq\beta$, where we again use \cref{eq:dynamical constraint for an insulator} with estimates independent of $\lambda$:
	
	\begin{align*}
		\int_{\Delta'}\norm{P'_\lambda(P_\lambda-P'_\lambda)_{,\alpha} P_{\lambda,\beta})}_1 \dif{\lambda}&\leq\int_{\Delta'}\norm{(P_\lambda-P'_\lambda)_{,\alpha} P_{\lambda,\beta})}_1 \dif{\lambda} \\
		&\leq \int_{\Delta'}\sum_{xyz}\norm{((P_\lambda-P'_\lambda)_{,\alpha})_{xy}}\norm{( P_{\lambda,\beta})_{yz}} \dif{\lambda} \\
		&\leq \sum_{xyz}C(1+\norm{y-z})^{-\alpha}|a(y)|^{-1}(1+|y_\beta|)^{-\alpha}\int_{\Delta'}\norm{((P_\lambda-P'_\lambda)_{,\alpha})_{xy}}  \dif{\lambda}  
	\end{align*}
	
	Now we have $\norm{(A_{,\alpha})_{xy}}=\norm{[\Lambda_\alpha,A]_{xy}}=\norm{(\Lambda(x_\alpha)-\Lambda(y_\alpha))A_{xy}}$. We can now invoke \cite[Proof of Lemma 2]{Graf_Shapiro_2018_1D_Chiral_BEC} to bound $|\Lambda(x_\alpha)-\Lambda(y_\alpha)|\leq C_\Lambda(1+|x_\alpha-y_\alpha|)^{\mu}(1+\frac{1}{2}|x_\alpha|)^{-\mu}$ to get, using \cref{prop:continuity of Fermi projections on average} again:
	
	\begin{align*}
	\int_{\Delta'}\norm{P'_\lambda(P_\lambda-P'_\lambda)_{,\alpha} P_{\lambda,\beta})}_1 \dif{\lambda}&\leq \sum_{xyz}C(1+\norm{y-z})^{-\alpha}|a(y)|^{-1}(1+|y_\beta|)^{-\alpha}\times\\&\qquad\times C_\Lambda(1+|x_\alpha-y_\alpha|)^{\mu}(1+\frac{1}{2}|x_\alpha|)^{-\mu}\int_{\Delta'}\norm{(P_\lambda-P'_\lambda)_{xy}}  \dif{\lambda} \\
	&\leq \sum_{xyz}C(1+\norm{y-z})^{-\alpha}|a(y)|^{-1}(1+|y_\beta|)^{-\alpha} \times\\
	&\qquad\times C_\Lambda(1+|x_\alpha-y_\alpha|)^{\mu}(1+\frac{1}{2}|x_\alpha|)^{-\mu}\times\\
	&\qquad\times\frac{2^{2-s/2}D_{0,\alpha} Q}{\pi s}d_\ell(H,H')^{s/2} |b(y)|^{-1/2}(1+\norm{x-y})^{-\alpha/2}\,.
	\end{align*}
	
	This last expression is unfortunately very long but the point is (when the dust settles) that it really is just a summable expression after the triple sum, so that we again get a constant times $d_\ell(H,H')^{s/2}$, which we can make as small as we like (there will still be a $\coth$ dependence on $d_\ell(H,H')$ coming from the sum, as in \cref{eq:estimate of operator norm by local metric}, and similarly that dependence approaches $d_\ell(H,H')\to0$). The last derivative term is dealt with in the same manner, and we find our result.
\end{proof}

In concluding this proof of continuity of $\Chern$, we compare it to the \emph{probabilistic} proof of \cite[Proposition 5.2]{PRODAN20161150}. In short, the latter proof shows that if $[0,1]\ni t\mapsto H(t)$ is a family of random ergodic Hamiltonians with $t$ the parameter of deformations, then $t\mapsto\EE[\Chern(H(t))]$ is locally constant. Since we know that $\EE[\Chern(H(t))]$ is almost-surely equal to $\Chern(H(t))$ (by Birkhoff), we conclude that almost-surely, if $|t-s|$ is small, $\Chern(H(t))=\Chern(H(s))$ (in the setting of \cite{PRODAN20161150}, the family $t\mapsto H(t)$ varies within one and the same random probability space so that it makes sense to compare random configurations at different values of $t$).

In contrast, \cref{thm:Chern number is continuous} shows that if the mobility gap property holds for two given realizations (independently) of near-by Hamiltonians (as measured by $d_\ell$), then their respective Chern numbers are necessarily equal. Hence, now we know it is impossible to toss some coins in a very lucky way in the laboratory and avoid the local constancy of the Chern number, and we know that the key topological property, namely the local constancy of the Chern number, is unrelated, and does not rely on covariance.

\begin{rem}
	In \cite{Shapiro_Tauber_BEC}, it is shown that the topological invariant for mobility-gapped Floquet 2D systems with no symmetry is invariant under selection of the logarithm branch cut within the mobility gap. This means that it might be possible for the proof above to be adapted for such systems. Indeed, apparently the crucial ingredients are statements such as \cref{prop:EGS_Chern_const_in_mob_gap} and a rewriting of the invariant in terms of contour integrals on resolvents, which allows for resolvent identities to be used. Part of \cite[Theorem 2.1]{Shapiro_Tauber_BEC} is the analog of \cref{prop:EGS_Chern_const_in_mob_gap}, though it is only through \cite[Theorem 2.6]{Shapiro_Tauber_BEC} that this is really established. Coincidentally an analog of \cref{prop:EGS_Chern_const_in_mob_gap} is precisely what we do not have for the chiral 1D systems studied in \cite{Graf_Shapiro_2018_1D_Chiral_BEC}.
	
	In contrast, there is little hope to control the local constancy of $\ZZ_2$ invariants associated to time-reversal invariant systems using the approach presented here, since an \emph{even} jump between two invariants should be allowed by continuous deformations.
\end{rem}

\section{Signatures of localization}\label{sec:signatures of localization}
In this section we justify the somewhat awkward \cref{def:insulator}.
Let $H\in\BLO{\HH}$ be a given random (ergodic) Hamiltonian. Let $\Delta\in\Open{\RR}$ be a given bounded interval. The fractional moments condition on $\Delta$ \cite[Lemma 2.1]{aizenman1993} says that for Lebesgue-almost-all $E\in\Delta$ there is some fraction $s_E\in(0,1)$ and constants $C_E<\infty,\mu_E>0$ such that \begin{align}
	\sup_{\eta\neq0}\EE[\norm{G(x,y;E+\ii\eta)}^{s_E}]\leq C_E \ee^{-\mu_E\norm{x-y} }\qquad(x,y\in\ZZ^d)\,.
\label{eq:the FMC}
\end{align}
Here $\EE$ is the disorder averaging and the other notation symbols are as in the preceding sections.

A further condition \cite[Eq. (4)]{Graf1994} that does not seem to follow automatically from \cref{eq:the FMC} (see \cite{Rajinder_Mavi_Jeffrey_Schenker2018}), but rather requires more input from $H$ is that for all $E\in\Delta$ there are constants $C_E<\infty,\mu_E>0$ (here and below, these constants are different for each constraint) such that \begin{align}
\sup_{\eta\neq0}\eta\EE[\norm{G(x,y;E+\ii\eta)}^2]\leq C_E \ee^{-\mu_E\norm{x-y} }\qquad(x,y\in\ZZ^d)\,.
\label{eq:the second moment condition}
\end{align}
Two very important consequences of these two conditions for topological insulators appeared in \cite{Aizenman_Graf_1998}. \cref{eq:the FMC} was shown to imply that for all $E\in\Delta$, there are constants $C_E<\infty,\mu_E>0$ such that \begin{align}
	\EE[\norm{\chi_{(-\infty,E)}(H)_{xy}}]\leq C_E \ee^{-\mu_E\norm{x-y} }\qquad(x,y\in\ZZ^d)\,.
	\label{eq:decay of Fermi projection}
\end{align}
\cref{eq:the second moment condition} in turn was shown to imply that there are constants $C<\infty,\mu>0$ such that \begin{align}
\EE[\sup_{f\in B_1(\Delta)}\norm{f(H)_{xy}}]\leq C \ee^{-\mu\norm{x-y} }\qquad(x,y\in\ZZ^d)\,.
\label{eq:decay of the Borel bounded functional calculus}
\end{align}
With $B_1(\Delta)$ as in \cref{def:insulator}. We note that since $\chi_{(-\infty,E)}\in B_1(\Delta)$, \cref{eq:decay of the Borel bounded functional calculus} implies \cref{eq:decay of Fermi projection}. Also, this implies (see e.g. \cite[Prop. A1]{Shapiro_Tauber_BEC}) that almost-surely, there is a (random) constant $C<\infty$ and (deterministic) $\mu>0,a\in\ell^1(\ZZ^d)$ such that \begin{align}
\sup_{f\in B_1(\Delta)}\norm{f(H)_{xy}}\leq C |a(x)|^{-1}\ee^{-\mu\norm{x-y} }\qquad(x,y\in\ZZ^d)\,.
\label{eq:deterministic decay of the Borel bounded functional calculus}
\end{align} as in, e.g., \cite[Eq. (1.2)]{EGS_2005}. 

Moreover, the Kubo formula \cref{eq:the Chern number} clearly shows that \cref{eq:deterministic decay of the Borel bounded functional calculus} implies the longitudinal DC conductivity is finite and zero, hence it makes sense to include it as the first part of the definition of an insulator,  \cref{def:insulator}. The second part, which concerns the finite degeneracy of localized eigenvalues within the mobility gap, is again a standard almost-sure consequence of models where localization is established \cite{Simon94} and we include it here for technical reasons, following \cite{EGS_2005}.

\begin{proof}[Proof of \cref{lem:Energy avg of Greens function decays uniformly in height from real axis}]
	First we note that via the RAGE theorem \cref{def:insulator} implies that $H$ has pure point spectrum within $\Delta$. 
	
	In \cite[Proof of Lemma 4]{EGS_2005}, it is proven that \cref{def:insulator} implies that $H$ has a SULE eigenbasis within $\Delta$, in the sense of \cite[eq-n (7.1)]{delRioJitomirskayaLastSimon1996}, i.e., we will use that there is some set of normalized eigenvectors $\Set{\psi_n}_{n\in\NN}$ with eigenvalues $\Set{\lambda_n}_{n\in\NN}\subseteq\Delta$ such that there is a constant $\mu>0$ such that for any $\ve>0$  there is some $C_\ve<\infty$ with $$\|\psi_n(x)\|\leq C_\ve \exp(-\mu\|x-x_n\|+\ve\|x_n\|)\qquad(x\in\ZZ^d)$$ for some $x_n\in\ZZ^d$ (the \emph{localization center} of $\psi_n$). Furthermore, in the proof of \cite[Corollary 7.3]{delRioJitomirskayaLastSimon1996}, it is shown that $\|x_n\|\geq\frac{1}{3}n^{\frac{1}{2}}-C_0$ for some $C_0<\infty$.
	
	Let $\Delta''$ be an open interval such that $\Delta'\subsetneq\Delta''\subsetneq\Delta$, and pick some smooth $\vf:\RR\to[0,1]$ such that $\left.\vf\right|_{\Delta^c}=0$ and $\left.\vf\right|_{\Delta'}=1$.
	
	Then for all $z\in\CC$ such that $\Re\{z\}\in\Delta'$, $\RR\ni\lambda\mapsto(1-\vf(\lambda))(\lambda-z)^{-1}$ is a smooth function with compact support and hence by the Helffer-Sj\"ostrand smooth functional calculus \cite[Appendix A]{Elbau_Graf_2002}, $(\Id-\vf(H))R(z)$ has position-basis matrix elements with  any rate polynomial off-diagonal decay, uniformly as $\Im\{z\}\to0$.
	
	Hence we concentrate on the off-diagonal decay of matrix elements of the operator $\vf(H)R(z)$. Since $\im(\chi_\Delta(H))$ is spanned by $\Set{\psi_n}_{n\in\NN}$ and $\supp(\vf)\subseteq\Delta$, we may write $$ \vf(H)R(z) = \sum_{n\in\NN} \frac{\vf(\lambda_n)}{\lambda_n-z}\psi_n\otimes\psi_n^\ast $$ so that using $(\sum_n a_n)^s \leq \sum_n a_n^s$, valid for all $s\leq1$, we get (for some $C'_s<\infty$) \begin{align*} \int_{E\in\Delta'}\|(\vf(H)R(E+\ii\eta))_{xy}\|^s\dif{E} &\leq \int_{E\in\Delta'}\sum_{n\in\NN} |\lambda_n-E-\ii\eta|^{-s}\|\psi_n(x)\|^s\|\psi_n(y)\|^s\dif{E}\\
	&\leq C'_s \sum_{n\in\NN} C_\ve^{2s} \exp(-s\mu(\|x-x_n\|+\|y-x_n\|)+2s\ve\|x_n\|)\\
	&\leq C''_{s,\ve} \exp(-\frac{1}{2}s\mu(\|x-y\|))\sum_n \exp(-3s\ve(\|x_n\|-\|x\|)+2s\ve\|x_n\|) \\
	& \leq C''_{s,\ve} \exp(-\frac{1}{2}s\mu(\|x-y\|)+3s\ve\|x\|)\sum_n \exp(-s\ve\frac{1}{3}n^{\frac{1}{2}}+3s\ve C_0) \\
	& \leq C'''_{s,\ve} \exp(-\frac{1}{2}s\mu(\|x-y\|)+3s\ve\|x\|)\,.  \end{align*}
	
	Combining the two estimates using $1 = \vf + (1-\vf)$ yields the result.
\end{proof}

\paragraph{Further consequences from the probabilistic fractional moment condition.}

Now we present the following additional \emph{deterministic} consequence of \cref{eq:the FMC}, which however--this is \cref{lem:Energy avg of Greens function decays uniformly in height from real axis}--is also a consequence of \cref{eq:deterministic decay of the Borel bounded functional calculus} together with the assumption of uniform finite degeneracy of the localized eigenvalues, as in \cref{def:insulator}. Hence, we include the following result mainly for general interest, but it is not strictly necessary for the logical flow.
\begin{prop}
	If $H$ is localized on $\Delta$ in the sense of \cref{eq:the FMC}, then almost-surely, for any $\alpha\in\NN$, there is a (random) constant $C_\alpha<\infty$ and (deterministic) $a\in\ell^1(\ZZ^d)$ such that \begin{align}
		\sup_{\eta\neq0}\int_{\Delta'}\norm{G(x,y;\cdot+\ii\eta)}^s \leq C_\alpha \frac{1}{|a(x)|}(1+\|x-y\|)^{-\alpha}\qquad(x,y\in\ZZ^d)
		\label{eq:energy average of fractional power of G is poly det decaying}
	\end{align} where $\Delta'\subset\Delta$ is a compact sub-interval. 
	\label{prop:Energy avg of Greens function decays uniformly in height from real axis}
\end{prop}

To calibrate \cref{eq:energy average of fractional power of G is poly det decaying}, let us first make the
\begin{rem}
	Regardless of the status of localization for $H$, for \emph{all} disorder configurations, for any fixed interval $I\subset\RR$, $s\in(0,1)$ and $x,y\in\ZZ^d$, \begin{align}\sup_{\eta\neq0}\int_I\norm{G(x,y;\cdot+\ii\eta)}^s<\infty\end{align} as shown in \cite[Eq. (8.2)]{AizenmanWarzel2016}.
	\label{rem:a-priori bound for energy avg}
\end{rem}

To prove \cref{prop:Energy avg of Greens function decays uniformly in height from real axis}, we start by making the weaker statement about what happens not in $\sup_{\eta\neq0}$ but only in the $\lim_{\eta\to0^+}$.

\begin{lem}
	If $H$ obeys the fractional moment condition on $\Delta$, \cref{eq:the FMC}, then for any compact sub-interval $\Delta'\subset\Delta$, almost-surely, there is a (random) constant $C<\infty$ and (deterministic) $a\in\ell^1(\ZZ^d),\mu>0,s\in(0,1)$ such that \begin{align}
	\limsup_{\eta\to0^+}\int_{\Delta'}\norm{G(x,y;\cdot+\ii\eta)}^s \leq C |a(x)|^{-1}\ee^{-\mu\norm{x-y}}\qquad(x,y\in\ZZ^d)
	\label{eq:energy average of fractional power of G is exp det decaying in limit to real axis}
	\end{align}
	\label{lem:energy average of fractional power of G is exp det decaying in limit to real axis}
\end{lem}
\begin{proof}
	Relating to the notations above \cref{eq:the FMC}, let $s:=\min(\Set{s_E|E\in\Delta'})$, $$ C':=|\Delta'|\max(\Set{C_{E}^{\frac{s}{s_{E}}}|E\in\Delta'}),\quad \mu:=\frac{s}{2\max(\Set{s_E|E\in\Delta'})}\min(\Set{\mu_E|E\in\Delta'}) $$ and pick \emph{any} $a\in\ell^1(\ZZ^d)$. By \cref{rem:a-priori bound for energy avg} we know that $\limsup_{\eta\to0^+}\int_{\Delta'}\norm{G(x,y;\cdot+\ii\eta)}^s<\infty$, so that we may use the dominated convergence theorem in the second line below. Fatou's lemma is applied in the first line, Fubini in the third, and we have \begin{align*}
		&\qquad\star:=\EE[\sum_{x,y\in\ZZ^d}\limsup_{\eta\to0^+}\int_{\Delta'}\norm{G(x,y;\cdot+\ii\eta)}^s|a(x)|\ee^{+\mu\norm{x-y}}]  \\ &\leq \liminf_{\Lambda\to\ZZ^d}\sum_{x,y\in\Lambda}\EE[\limsup_{\eta\to0^+}\int_{\Delta'}\norm{G(x,y;\cdot+\ii\eta)}^s|a(x)|\ee^{+\mu\norm{x-y}}] \\ &= \liminf_{\Lambda\to\ZZ^d}\sum_{x,y\in\Lambda}\limsup_{\eta\to0^+}\EE[\int_{\Delta'}\norm{G(x,y;\cdot+\ii\eta)}^s|a(x)|\ee^{+\mu\norm{x-y}}] \\ &= \liminf_{\Lambda\to\ZZ^d}\sum_{x,y\in\Lambda}\limsup_{\eta\to0^+}\int_{\Delta'}\EE[\norm{G(x,y;\cdot+\ii\eta)}^s]|a(x)|\ee^{+\mu\norm{x-y}}
	\end{align*}
	Now we note that if \cref{eq:the FMC} holds for some $s_E\in(0,1)$, then the same holds with $s_E$ replaced by $\sigma$ for all $\sigma\in(0,s_E)$ by Jensen's inequality: \begin{align*}
		\sup_{\eta\neq0}\EE[\norm{G(x,y;E+\ii\eta)}^\sigma]\leq C_{E}^{\frac{\sigma}{s_{E}}}\ee^{-\frac{\sigma}{s_{E}}\mu_{E}\norm{x-y}}\qquad(x,y\in\ZZ^d)\,.
	\end{align*} If we restrict to $E\in\Delta'$ then the RHS is bounded by $\frac{1}{|\Delta'|}C'\ee^{-2\mu\norm{x-y}}$. Hence \begin{align*}\star&\leq\liminf_{\Lambda\to\ZZ^d}\sum_{x,y\in\Lambda}\limsup_{\eta\to0^+}\int_{\Delta'}\frac{1}{|\Delta'|}C'\ee^{-2\mu\norm{x-y}}|a(x)|\ee^{+\mu\norm{x-y}} \\ &= C'\sum_{x,y\in\ZZ^d}\ee^{-\mu\norm{x-y}}|a(x)| \\ &= C' \norm{a}_{\ell^1}\sum_{x\in\ZZ^d}\ee^{-\mu\norm{x}} \\ & < \infty\,.\end{align*} Hence there is some random $C<\infty$ such that $$ \sum_{x,y\in\ZZ^d}\limsup_{\eta\to0^+}\int_{\Delta'}\norm{G(x,y;\cdot+\ii\eta)}^s|a(x)|\ee^{+\mu\norm{x-y}} \leq C $$ and so also $$ \limsup_{\eta\to0^+}\int_{\Delta'}\norm{G(x,y;\cdot+\ii\eta)}^s \leq C |a(x)|^{-1}\ee^{-\mu\norm{x-y}} \qquad (x,y\in\ZZ^d)$$ which is what we wanted to prove.
\end{proof}

Our next task is to upgrade the $\limsup$ to a $\sup$. We first have to establish a certain subharmonicity: 
\begin{lem}
	With $\mathbb{C}_{+}:=\Set{z|\Im\left\{ z\right\} >0}$ we have that for any $a>0,s\in(0,1)$,
	\begin{align*}
		\mathbb{C}_{+}\ni z & \mapsto  \int_{\Re\left\{ z\right\} -a}^{\Re\left\{ z\right\} +a}\norm{G\left(x,y;\cdot+\ii\Im\left\{ z\right\} \right)}^{s}=:\varphi_{x,y,a,s}\left(z\right)\in\mathbb{R}
	\end{align*}
	is a subharmonic function.
	\begin{proof}
		First note that $\mathbb{C}_{+}\ni z\mapsto G\left(x,y;z \right)\in \Mat_{N}\left(\mathbb{C}\right)$
		is analytic, so that $\mathbb{C}_{+}\ni z\mapsto\norm{G\left(x,y;z\right)}^{s}=:g\left(z\right)$
		is \emph{subharmonic}. To verify that $\varphi$ is subharmonic, we
		pick any $z\in\mathbb{C}_{+}$ and $r>0$ such that $\overline{B_{r}\left(z\right)}\subseteq\mathbb{C}_{+}$
		(so $r<\Im\left\{ z\right\} $) and aim to prove $\varphi\left(z\right)\leq\frac{1}{2\pi}\int_{0}^{2\pi}\varphi\left(z+r\ee^{i\theta}\right)\dif{\theta}$:
		\begin{align*}
			\varphi\left(z\right) & \equiv  \int_{\Re\left\{ z\right\} -a}^{\Re\left\{ z\right\} +a}g\left(\lambda+\ii\Im\left\{ z\right\} \right)\dif{\lambda}\\
			&   \left(\text{Subharmonicity of }g\right)\\
			& \leq  \int_{\Re\left\{ z\right\} -a}^{\Re\left\{ z\right\} +a}\frac{1}{2\pi}\int_{\theta=0}^{2\pi}g\left(\lambda+\ii\Im\left\{ z\right\} +r\ee^{\ii\theta}\right)\dif{\theta}\dif{\lambda}\\
			&   \left(\text{Fubini}\right)\\
			& \leq  \frac{1}{2\pi}\int_{\theta=0}^{2\pi}\int_{\Re\left\{ z\right\} -a}^{\Re\left\{ z\right\} +a}g\left(\lambda+\ii\Im\left\{ z\right\} +r\ee^{\ii\theta}\right)\dif{\lambda}\dif{\theta}\\
			&   \left(\lambda':=\lambda+r\cos\left(\theta\right)\right)\\
			& =  \frac{1}{2\pi}\int_{\theta=0}^{2\pi}\int_{\Re\left\{ z\right\} -a+r\cos\left(\theta\right)}^{\Re\left\{ z\right\} +a+r\cos\left(\theta\right)}g\left(\lambda'+\ii\left(\Im\left\{ z\right\} +r\sin\left(\theta\right)\right)\right)\dif{\lambda}\dif{\theta}\\
			&   \left(\text{Use definition of }\varphi\right)\\
			& =  \frac{1}{2\pi}\int_{\theta=0}^{2\pi}\varphi\left(\Re\left\{ z\right\} +r\cos\left(\theta\right)+\ii\left(\Im\left\{ z\right\} +r\sin\left(\theta\right)\right)\right)\dif{\theta}\\
			& =  \frac{1}{2\pi}\int_{\theta=0}^{2\pi}\varphi\left(z+r\ee^{\ii\theta}\right)\dif{\theta}\,.
		\end{align*}
	\end{proof}
\end{lem}
\begin{proof}[Proof of \cref{prop:Energy avg of Greens function decays uniformly in height from real axis}]

Since $\varphi_{x,y,a,s}$ is a sub-harmonic function which decays in $\norm{x-y}$ above the real axis (via the Combes-Thomas estimate) and whose $\limsup$ on the real axis decays in $\norm{x-y}$ via \cref{lem:energy average of fractional power of G is exp det decaying in limit to real axis}, we find our result using either \cite[Theorem 4.2]{Aizenman2001} or \cite[Proposition 25]{Graf_Shapiro_2017}.
\end{proof}

\medskip
\noindent\textbf{Acknowledgements:} It is a pleasure to thank Gian Michele Graf for many useful discussions. I am grateful to Jeffrey Schenker for \cref{lem:Energy avg of Greens function decays uniformly in height from real axis}, which allows one to avoid \cref{prop:Energy avg of Greens function decays uniformly in height from real axis} and hence greatly simplify \cref{def:insulator}. This research is supported in part by Simons Foundation Math + X Investigator Award \#376319 (Michael I. Weinstein).

\begingroup
\let\itshape\upshape
\printbibliography
\endgroup
\end{document}